\documentclass[reqno]{amsart} \usepackage{amscd}
\usepackage{epsf}
\newtheorem{theorem}{Theorem}[section]
\newtheorem{proposition}[theorem]{Proposition}
\newtheorem{lemma}[theorem]{Lemma}

\theoremstyle{remark} 
\numberwithin{equation}{section}
\newcommand{\field}[1]{\ensuremath{\mathbb{#1}}}
\newcommand{\CC}{\field{C}}
\newcommand{\HH}{\field{H}}
\newcommand{\RR}{\field{R}}
\newcommand{\ZZ}{\field{Z}}
\begin{document}

\title[Geometric prequantization ]{Geometric prequantization of a modified 
Seiberg-Witten moduli space in $2$ dimensions}

\author{Rukmini Dey}

\begin{abstract}
In this paper we consider a dimensional reduction of slightly modified 
 Seiberg-Witten equations, the modification being a different choice of the
Pauli matrices which go into defining the equations. We get interesting equations with a Higgs field , spinors and a connection. We show interesting solutions
of these equations. Then we go on to show a family of symplectic structures on the moduli space of these equations which can be geometrically prequantized using the Quillen determinant line bundle.
\end{abstract}

\maketitle

\section{Introduction}

It is important to study the dimensional reductions of gauge theories for they sometimes possess beautiful symplectic or hyperK\"{a}hler structures which can be geometrically quantised, ~\cite{D1}, ~\cite{D2}. It is hoped that  these Hilbert spaces of the quantizations could be used to produce invariants of $3$ or $4$
dimensional manifolds as in perhaps ~\cite{Wi}, ~\cite{G}, or perhaps could be used in Gromov-Witten theory ~\cite{T}.

In this paper we modify the Seiberg-Witten equations in ${\RR}^4$ by choosing different $I$, $J$ and $K$ from the standard one and dimensionally reduce the 
equations to ${\RR}^2$ and then finally patch them on the Riemann surface , much in the way it was done in ~\cite{D3} though  the equations look different from those in ~\cite{D3}. Then we show that the moduli space 
 is non-empty and in fact there are intersting solutions. Then we show there is a family of symplectic structures and geometrically prequantize them along  the ideas of ~\cite{D1} and ~\cite{D2}.

   We must mention that in ~\cite{D3} and ~\cite{D4}, the author had attempted to geometrically quantize 
the dimensional reduction of the Seiberg-Witten equations with a Higgs field.
But there are some mistakes in these two papers which are to be rectified.  The present paper could be  thought of as  a modification and extenstion of  the work done in ~\cite{D3} and ~\cite{D4}

Geometric prequatization has been described in the introductions of 
~\cite{D1} and ~\cite{D2}. Let us just mention that it involves constructing
a line bundle on the moduli space whose curvature is a symplectic form 
on the moduli space. In fact, we get a family of symplectic forms, parametrised
by $\psi_0$, a section of a line bundle, which are quantised this way.
For each of them we construct a Quillen determinant line bundle whose curvature is that symplectic form. Note that topologically all these line bundles are equivalent, since their Chern class is intergral and doesnot vary. However holomorphically they may be distinct.

The equations in the Hitchin system involved a connection $A$ and a 
Higgs field $\Phi$.
In the vortex equation, a connection $A$ and one other field $\Psi$ appeared. The 
equations we are dealing with in this paper are more complex and involve a connection $A$, a Higgs field $\Phi$ and two other fields $\psi_1$ and $\psi_2$. It would be interesting to find  an algebraic geometric interpretation of the latter moduli space  like that of  
the Hitchin systems and the vortex moduli space.   Also, as in the case of 
Chern-Simons theory and flat connections, ~\cite{Wi}, it would be interesting to find a Lagrangian theory in $3$-dimensions whose  quantization will lead naturally to the prequantization described in this paper. The Hilbert space of the quantization of the moduli space of flat connections turned out to be the space of
conformal blocks in a certain conformal field theory, ~\cite{ND}. One could also try to answer the analogous question in the three cases, namely the Hitchin system, vortex and the present case. As a result one might get $3$-manifold invariants as in ~\cite{Wi}  and ~\cite{G}.

\section{Dimensional Reductions of the Seiberg-Witten equations}

In this section we dimensionally reduce the modified Seiberg - Witten
equations on ${\RR }^4$ to  ${\RR }^2$ and define them over a
compact Riemann surface $M$.

\subsection{The Seiberg-Witten equations on ${\RR }^4$: }
This is a brief description of the Seiberg-Witten equations on
${\RR}^4$,~\cite{S},~\cite{Ak}, ~\cite{M} .

Identify ${\RR }^4 $  with the quaternions ${\HH}$ (coordinates $x = (x_1, x_2, x_3, x_4) $ identified with $\zeta=(\zeta_1, \zeta_2, \zeta_3, \zeta_4)$) and
let $\{e_i, i=1,2,3,4 \}$ be a  basis for ${\HH}$.
Fix the constant spin structure
$\Gamma : {\HH} = T_x {\HH} \rightarrow {\CC}^{4 \times 4}$, given by
$\Gamma (\zeta)  = \left[ \begin{array}{cc}
0              & \gamma (\zeta) \\
\gamma(\zeta) ^{*}  & 0
\end{array} \right], $
where $\gamma (\zeta) = \left[ \begin{array}{cc}
\zeta_1 - i \zeta_2   &\zeta_3 - i \zeta_4 \\
-\zeta_3 + i \zeta_4  & \zeta_1 - i \zeta_2
\end{array} \right]. $
 Thus $\gamma (e_1) = Id$, $\gamma (e_2) = I$, $\gamma (e_3) = J$,
$\gamma (e_4) = K$ where
$$I =  \left[ \begin{array}{cc}
-i & 0 \\
0 & -i
\end{array} \right], J=  \left[ \begin{array}{cc}
0              & 1 \\
-1             & 0
\end{array} \right], K =  \left[ \begin{array}{cc}
0              & -i \\
i             &  0
\end{array} \right], $$

Note: The choice of $I$, $J$ and $K$ is not standard. They donot satisfy the quaternionic algebra. This is our point of deviation from the Seiberg-Witten theory. For the standard choice see ~\cite{D3}.

Recall that $Spin^{c}({\RR }^4) = (Spin ({\RR }^4) \times S^1) / {\ZZ}_2$.
Since $Spin ({\RR }^4)$ is a double cover of $SO(4)$, a $spin^c$ -
connection involves a connection $\omega$ on $T{\HH}$ and a connection
$ A = i\sum \limits_{j=1}^{4} A_j d x_j  \in \Omega^1 ({\HH}, i
{\RR})$ on
the characteristic line bundle $ {\HH} \times {\CC}$ which arises
from the $S^1$ factor (see ~\cite{S}, ~\cite{M}, ~\cite{Ak} for more
details). We set $\omega = 0$, which is equivalent to choosing the
covariant derivative  on the trivial tangent bundle to be $d$.  This is
legitimate since we are on ${\RR }^4$.   The  curvature $2$-form of the
connection $A$ is given by
$F(A) = d A    \in \Omega^2({\HH}, i{\RR})$. Consider the covariant
derivative acting on $\Psi \in C^{\infty}({\HH}, {\CC}^2)$ (the
positive spinor on ${\RR }^4$)  induced by the connection $A$ on
${\HH} \times {\CC}:$ $\nabla_j  \Psi = (\frac{\partial }{\partial
x_j} + i A_j) \Psi.  $  Then according to ~\cite{S}, the
Seiberg-Witten equations for $(A, \Psi)$ on ${\RR }^4$ are
equivalent to the equations:

$(SW1):$ $\nabla_1 \Psi = I \nabla_2
\Psi + J \nabla_3 \Psi+ K  \nabla_4 \Psi,$

$(SW2a):$  $F_{12} + F_{34} =  \frac{1}{2} \Psi^{*} I \Psi =
\frac{-i}{2}(|\psi_1 |^2 + |\psi_2 |^2) \stackrel{\cdot}{=}\frac{1}{2}\eta_1,$

$(SW2b):$ $ F_{13} + F_{42} = \frac{1}{2} \Psi^{*} J \Psi = -i (Im \psi_1 \psi_2) \stackrel{\cdot}{=}\frac{1}{2} \eta_2, $

$(SW2c):$ $ F_{14} + F_{23} =
\frac{1}{2} \Psi^{*} K \Psi =  - Im \psi_1 \psi_2
\stackrel{\cdot}{=}\frac{1}{2} \eta_3 $

where $\Psi = \left[ \begin{array}{cc}
\psi_1 \\
\bar{\psi}_2
\end{array} \right],$
where by our convention $F_{12} = i (\partial_2 A_1 - \partial_1 A_2)$ etc.

\subsection{Dimensional Reduction to ${\RR}^2$ }: Using the same method of
dimensional reduction as  in ~\cite{H}, we get the general form of the
reduced equations which contain the so-called Higgs field. Namely, impose the
condition that none of the $A_i$'s and  $\Psi$ in $(SW1)$ and $(SW2)$ depend
on $x_3$ and $x_4$, i.e. $ A_i = A_i(x_1, x_2) $, $\Psi = \Psi(x_1, x_2)$ and
 set $\phi_1 = -i A_3$ and $\phi_2 = -i A_4$. The $(SW2)$
 equations reduce to the following system on ${\RR }^2$, $ F_{12}
=\frac{1}{2} \eta_1, $ and two other equations which is as follows
 $ \frac{{\partial} (\phi_1 - i \phi_2 )}{ \partial z } =
 \frac{1}{2}(\eta_2 -i\eta_3) = 0, $ where
 $\frac{\partial}{\partial z} =
\frac{1}{2}(\frac{\partial}{\partial x_1} -i
\frac{\partial}{\partial x_2 }). $ 
This is because  $0= F_{13} + F_{14} - i(F_{14} + F_{23}) = \partial_1 \phi_1 - \partial_2 \phi_2 -i (\partial_1 \phi_2 + \partial_2 \phi_1) = (\partial_1 - i \partial_2 ) (\phi_1 - i \phi_2). $

Setting $( \phi_1 - i \phi_2) =
\bar{\phi} $  and recalling $dx_2 \wedge dx_1 = -i dz \wedge d \bar{z}$ we  rewrite the reduction of
$(SW2)$ as the following two equations,

\begin{eqnarray*}
(1) {\rm \; \;}  F(A) &=& \frac{-i}{2} ( |\psi_1|^2 + |\psi_2|^2 ) dx_2 \wedge d x_1 \\
            &=&  \frac{i}{2} ( |\psi_1|^2 + |\psi_2|^2 ) i dz \wedge d {\bar{z}}, 
\end{eqnarray*}

$$(2) {\rm \; \;}  \partial \Phi^{0,1} = 0 $$

where
$ \Phi = \Phi^{1,0} + \Phi^{0,1} = \phi dz - \bar{\phi} d \bar{z} \in \Omega^1 ({\RR }^2,
{i\RR })$  and $ \psi_1 , \psi_2  \in C^{\infty} ({\RR}^2, {\CC}) $
are spinors on ${\RR }^2$. Next  consider the Dirac
equation   $(SW1)$:

$ \nabla_1 \psi - I \nabla_2 \psi - J \nabla_3 \psi - K \nabla_4 \psi = 0 $
which is rewritten as

$\left[ \begin{array}{cc}
\frac{\partial}{\partial x_1 } + iA_1 + i \frac{\partial}{\partial x_2} - A_2
 &   -iA_3  - A_4     \\
iA_3 + A_4  & \frac{\partial}{\partial x_1} + iA_1 + i \frac{\partial}{\partial
 x_2} - A_2 \end{array}
\right] \left[ \begin{array}{cc}
\psi_1 \\
\bar{\psi}_2
\end{array} \right]   = 0.$

Introducing $A^{1,0} = \frac{i}{2}(A_1 - i A_2) dz$ and
$A^{0,1} = \frac{i}{2} (A_1 + i A_2) d \bar{z}$ where the total connection
$A^{1,0} +  A^{0,1} =i (A_1 dx + A_2 dy)$, we can finally write it as

$$ (3) {\rm \; \;}  \left[ \begin{array}{cc} 2 (\bar{\partial} + A^{0,1}) &  \bar{\phi} d \bar{z}\\
  -\bar{\phi} d \bar{z} & 2 (\bar{\partial} + A^{0,1})
\end{array} \right] \left[ \begin{array}{cc}
 \psi_1 \\
\bar{\psi}_2
\end{array} \right] = 0  $$

We call equations $(1) - (3)$ as the dimensionally reduced
Seiberg-Witten equations over ${\CC}$.

\subsection{The Dimensionally Reduced Equations on a Riemann surface}

Let $M$ be a compact Riemann surface of genus $g$ with a
conformal metric $ds^2 = h^2 dz \otimes d \bar{z}$ and let $\omega
= i e^{2\sigma}h^2 d z \wedge d \bar{z}$ be a real form. Let  $L$ be a  line bundle with a Hermitian metric $H$.  Let $\psi_1,\bar{\psi_2}$ be sections of the line
bundle $L$ i.e., $\psi_1 \in \Gamma(M,L)$ and $\psi_2 \in \Gamma(M, \bar{L})$. 
$L$ has a Hermitian metric $H$ and thus we can define an inner product 
between two sections $\psi$ and $\tau$ as follows: $\psi=f e$, $\tau = g e$ where
$e$ is a  section of $L$
then $<\psi, \tau>_H = f \bar{g} <e, e>_H \in C^{\infty}(M).$ By abuse of notation we write 
$<\psi, \tau>_H = \psi H \bar{\tau}$.  This inner product will come in handy 
when defining the determinant line bundles. 
  The norm $|\psi |_H \in C^{\infty}(M) $ . Let $A^{1,0} + A^{0,1}$ be a unitary connection on $L$, i.e. $\overline{A^{1,0}}= - A^{0,1}$,   and $\Phi = \Phi^{1,0} + \Phi^{0,1} =  \phi dz -\bar{\phi} d \bar{z} \in
\Omega^1 (M , i{\RR })$. We will  assume that $\Psi = \left[
\begin{array}{c}
\psi_1\\
\bar{\psi}_2
\end{array} \right]$ is not identically zero. We can
rewrite the equations $(1) - (3)$ in an invariant form on $M$ as
follows:

 $$F(A) = i \frac{( |\psi_1 |_H ^2 + |\psi_2 |_H ^2 )}{2}
\omega, \leqno{(2.1)}$$

$$\partial \Phi^{0,1}  =  0, \leqno{(2.2)}$$

$$\left[ \begin{array}{cc} \bar{\partial} + A^{0,1} & \frac{1}{2} \bar{\phi} d\bar{z}  \\
 -\frac{1}{2} \bar{\phi} d \bar{z} & \bar{\partial} + A^{0,1}
\end{array} \right] \left[ \begin{array}{cc}
\psi_1 \\
\bar{\psi}_2
\end{array} \right] = 0.\leqno{(2.3)}$$

Let ${\mathcal C} = {\mathcal A} \times \Gamma (M, L \oplus L)
\times {\mathcal H}$ , where ${\mathcal A}$ is the space of
connections on a line bundle $L$, $\Gamma (M, L \oplus L) $ the
space of sections of the  bundle $L \oplus L$ and ${\mathcal H}$
be $\Omega^{1}(M, i {\RR})$, the space of Higgs fields. Then $(A, \Psi = \left[ \begin{array}{cc}
\psi_1 \\
\bar{\psi}_2
\end{array} \right], \Phi) \in {\mathcal C}.$ The gauge
group ${\mathcal G}$ which is locally  $ Maps (M, U(1))$ acts on ${\mathcal B}$ as
$(A, \Psi, \Phi) \rightarrow (A + u^{-1} du, u^{-1} \Psi, \Phi)$
and leaves the space of solutions to $(2.1) - (2.3) $ invariant.
There are no fixed points of this action. Because  a fixed point
would mean that there is a connection $A_0$ such that $A_0 +
u^{-1}du = A_0$ for all $u$ in the gauge group. This is not
possible. We assume throughout that $\Psi$ is not identically
zero. Note that we let $\sigma$ also vary.

By taking quotient of the space of solutions by the gauge group
we get the moduli space ${\mathcal N}$.

\begin{proposition}
The moduli space ${\mathcal N}$ is not empty for a compact (oriented) Riemann surface of 
genus $g>1$.
\end{proposition}

\begin{proof}

Let us take the line bundle $L$ to be the tangent bundle of a compact (oriented) Riemann 
surface of genus $g>1$.  We take the connection 
$A = A^{1,0} + A^{0,1}= \partial {\rm ln} (e^{\sigma} h) - \bar{\partial} {\rm ln} (e^{\sigma} h ),  $ ( ~\cite{GH}, page 77).  
Let us take $\Phi =0$. The second equation is solved naturally. 
The third equation becomes 
$$(\bar{\partial} + A^{0,1}) \psi_1 = 0$$ 
$$(\bar{\partial} + A^{0,1}) \bar{\psi}_2 = 0$$
where $A^{0,1} = \bar{\partial} ln (e^{\sigma} h)$ , etc. These two equations imply
$$\bar{\partial} {\rm ln} (e^{\sigma}h \psi_1)=0$$
$$ \bar{\partial} {\rm ln} (e^{\sigma}h \bar{\psi}_2)=0$$
or in otherwords, $ {\rm ln} (e^{\sigma}h \psi_1)= f(z)$ and $ {\rm ln} (e^{\sigma}h \bar{\psi}_2)= g (z).$
Thus we get a whole family of solutions 
$$\psi_1 = e^{f(z)} e^{-\sigma} h^{-1}$$
$$\bar{\psi}_2 = e^{g(z)} e^{-\sigma} h^{-1}$$

Next the first equation becomes
\begin{eqnarray*}
F(A) &=& d A = -\frac{1}{2} \Delta ln (e^{\sigma}h) d z \wedge d \bar{z}\\
     &=&  K(e^{\sigma}h) e^{2\sigma} h^2 dz \wedge d \bar{z}\\
     &=& \frac{i}{2} (|\psi_1|_H^2 + |\psi_2|_H^2) \omega\\
     &=& \frac{-1}{2} (|\psi_1|_H^2 + |\psi_2|_H^2) e^{2\sigma} h^2 dz \wedge d \bar{z}
\end{eqnarray*}
which implies that $K(e^{\sigma}h) = \frac{-1}{2} (|\psi_1|_H^2 + |\psi_2|_H^2)$.
This always has a solution $\sigma$ for a compact (oriented)  genus $g >1$ surfaces, see ~\cite{De}.
\end{proof}

\begin{proposition}
There exists global solutions with $\Phi \neq 0$ identically and $\psi_1$
and $\psi_2$ not equal to zero identically, on a compact oriented Riemann 
surface of genus $g >1$.
\end{proposition}

\begin{proof}
Let $X_1$ be a torus with a puncture where the puncture looks like a long cylinder with negative Gaussian curvature at the root of the cylinder  and zero Gaussian curvature at the end. Let $X_2$ be a long cylinder with zero Gaussian curvature everywhere. Let $X_3$ be an identical copy of $X_1$.
Consider a Riemann surface $X$ which is obtained by gluing $X_1$, $X_2$ and 
$X_3$ where the gluing occurs along the flat ends of the cylinders and $X_2$
is in the middle of $X_1$ and $X_3$. In other words, to $X_1$ we glue $X_2$ along the flat ends of the cylinders, and then to $X_2$ we glue $X_3$ along the flat ends of the cylinders.

\medskip

Let the line bundle $L$ be the tangent bundle on the surface $X$ which is flat 
on the middle cylinder  and non-flat at the torus ends.
Thus let $\psi_1$ and $\psi_2$ be defined on $X_1$ such that they are non zero 
on most of $X_1$ and decay very fast to zero at the end of the cylinder of $X_1$ such that the first derivatives are also zero at the end of the cylinder. 
By ~\cite{De2} (where we take $K_0$ to be negative near the root of the cylinder
and zero at the end of the cylinder) there is a solution $\sigma$ to the equation: $F(A) = K (e^{\sigma} h) e^{2\sigma} h^2 dz \wedge d \bar{z} = \frac{i}{2} (|\psi_1|_H^2 + |\psi_2|_H^2) \omega$ where recall $\omega = i e^{2\sigma} h^2 dz \wedge d \bar{z}$   i.e.  $K(e^{\sigma}h) =\frac{-1}{2} (|\psi_1|_H^2 + |\psi_2|_H^2) $ is negative everywhere on $X_1$ except on the cylinder where it is $K_0$,  which is negative at the root of the cylinder  and zero at the end of the cylinder. This is possible since $X_1$ is a hyperbolic Riemann surface with puncture, ~\cite{De2}. As mentioned before we take $\psi_1$ and $\psi_2$ to be decaying to zero fast enough at the end of the cylinder of $X_1$. This is possible by suitable choice of $f(z)$ and $g(z)$, (notation as in proposition $(2.1)$.  We take $\Phi =0$ on $X_1$.

\medskip

On $X_2$ we take  $\psi_1 = \psi_2 = 0$, i.e. $F(A) =0$ (which is possible since $X_2$ is a flat cylinder) and $\Phi^{0,1} = c(\bar{z}) d \bar{z} $ where $c(\bar{z})$ decays to zero fast enough so that its first derivatives are zero at the two ends
of the cylinder $X_2$.

\medskip

On $X_3$ we have a solution exactly as in $X_1$ with $\Phi =0$ but $\psi_1$
and $\psi_2$ non-zero. 

\medskip

In the two cylindrical regions where the three solutions are glued all of  $\Phi,$  $\psi_1$ and $\psi_2$ are zero and their derivatives are also zero (since $f$ and 
$g$ and $c$ decay to zero very fast) -- so that the equations which involve 
first derivatives are satisfied identically. 

Thus on $X$, a compact oriented Riemann surface with genus $g=2$, we have constructed a solution with $\Phi, \psi_1, \psi_2$ non-identically zero.
   
By repeating this process, we can get the result on any  genus $g>1$ surface
-- because we can construct these by adding torus with a  cylindrical
puncture to a genus $g-1$ surface with a long cylindrical puncture and on any of the hyperbolic pieces we can have solution as in $X_1$ and  $\Phi \neq 0$ in the middle flat cylinders. 
\end{proof}

\begin{proposition}
Let us consider the moduli space
${\mathcal N}.$
 Suppose $(A, \Psi, \Phi)$ is a point on the moduli space such
that $\Psi$ is not identically  $0$.  The (virtual) dimension  of 
${\mathcal N}$ is $2 g + 2 c_1(L) + 2$

If $\Phi = 0$ then $(i)$ if $\psi_1$ and $\bar{\psi}_2$   are not identicaly zero,
then the dimension is $2 c_1 (L) + 2 $ and $(ii)$ if $\psi_1 \equiv 0$ then
the dimension is $g + c_1(L) +1 $. 
\end{proposition}

\begin{proof}

To calculate the dimension of ${\mathcal N}$  let ${\mathcal S}$ be the
solution space to $(2.1)-(2.3)$.
Consider the tangent space $T_p {\mathcal S}$ at a point
  $p= (A, \Psi =\left[ \begin{array}{cc}
\psi_1 \\
\bar{\psi}_2
\end{array} \right] , \Phi) \in {\mathcal S},$ which is defined by the
 linearization of equations $(2.1)-(2.3)$. Let $X = (\alpha, \beta , \gamma) \in T_{p}({\mathcal S}) $, where $\alpha \in \Omega^{1}
(M, i {\RR })$ and $ \beta =\left[ \begin{array}{cc}
\beta_1 \\
\bar{\beta}_2
\end{array} \right]  \in \Gamma (M, L \oplus L), $ and $\gamma \in {\mathcal H}$. The
linearizations of the equations are as follows

$$d\alpha = i(  <\psi_1, \beta_1>_H + <\psi_2 , \beta_2 >_H)  \omega, \leqno{(2.1)^{\prime}}$$

$$\partial \gamma^{0,1} = 0 \leqno{(2.2)^{\prime}}$$

$$ \left[ \begin{array}{cc}
  \bar{\partial} + A^{0,1} & \frac{1}{2} \bar{\phi} d \bar{z}  \\
 -\frac{1}{2} \bar{\phi} d \bar{z} & \bar{\partial} + A^{0,1}
\end{array} \right] \left[
\begin{array}{c}
\beta_1 \\
\bar{\beta}_2
\end{array} \right]    +  \left[ \begin{array}{cc}
 \alpha^{0,1}  & \frac{1}{2}\gamma^{0,1}   \\
 -\frac{1}{2} \gamma^{0,1} & \alpha^{0,1} 
\end{array} \right] \left[
\begin{array}{cc}
\psi_1 \\
\bar{\psi}_2
\end{array} \right] = 0. \leqno{(2.3)^{\prime}} $$

Taking into account the  quotient by the gauge group ${\mathcal  G}$, we
arrive at the following  sequence ${\mathcal C}$
$$0 \rightarrow \Omega^0(M, i{\RR }) \stackrel{d_1}{\rightarrow}
\Omega^1 (M, i{\RR }) \oplus \Gamma (M, {\mathcal L}) \oplus {\mathcal H}
\stackrel{d_2}{\rightarrow} \Omega^2 (M, i{\RR }) \oplus \Omega^{2}(M, {\CC})
\oplus    V \rightarrow 0 ,$$
where  ${\mathcal L} = L \oplus L$, $V =  (L \otimes \Omega^{1,0}(M) ) \oplus  (L \otimes \Omega^{1,0}(M)) $,

$d_1 f = (df, -f \Psi, 0)$, $ d_2 (\alpha, \left[ \begin{array}{c}
\beta_1 \\
\bar{\beta}_2
\end{array} \right] , \gamma)   \stackrel{\cdot}{=} (A, B, C ), $

$A= d\alpha - i [<\psi_1, \beta_1>_H + < \psi_2, \beta_2>_H ] \omega \in \Omega^{2}(M, i {\RR})$

$B= \partial \gamma^{0,1}  \in \Omega^{2}(M,{\CC})$

$C=  \left[ \begin{array}{cc}
  \bar{\partial} + A^{0,1} & \frac{1}{2} \bar{\phi} d \bar{z}  \\
 -\frac{1}{2} \bar{\phi} d \bar{z} & \bar{\partial} + A^{0,1}
\end{array} \right] \left[
\begin{array}{c}
\beta_1 \\
\bar{\beta}_2
\end{array} \right]    +  \left[ \begin{array}{cc}
 \alpha^{0,1}  & \frac{1}{2}\gamma^{0,1}   \\
 -\frac{1}{2} \gamma^{0,1} & \alpha^{0,1} 
\end{array} \right] \left[
\begin{array}{cc}
\psi_1 \\
\bar{\psi}_2
\end{array} \right] \in V.$

[ Note that the Lie algebra of the gauge group acts locally like $(df, -f \Psi, 0)$ where $u= e^f$, $f$ is purely imaginary.] 

It is easy to check that $d_2 d_1 = 0$, so that this is a complex.
Clearly, $H^0({\mathcal C})$ $=$ $0$ ,  because if  $f \in $
 $ker(d_1),$ then $df =0$ and $f \Psi =0$, which  implies  $f = 0$ since we
are in the neighbourhood of a point where $\Psi \neq 0$.

The Zariski dimension of the moduli space is $dim H^1({\mathcal C})$
while the virtual dimension is
dim $H^1({\mathcal C}) -$ dim $H^2({\mathcal C})$,
and coincides with the Zariski dimension whenever $dim H^2 ({\mathcal C})$ is
zero (namely the smooth points of the solution space ~\cite{M}, page 66).
The virtual dimension is $ = dim H^1({\mathcal C}) - dim H^2({\mathcal C}) =$
index of ${\mathcal C} $.

To calculate the index of ${\mathcal C}$ , we consider the family
of complexes $({\mathcal C}^t, d^t)$, $0 \leq t \leq 1,$ where

$d_{1}^t = (df, -t f \Psi, 0)$,
$d_2 (\alpha, \left[ \begin{array}{cc}
\beta_1 \\
\beta_2
\end{array} \right] , \gamma)   \stackrel{\cdot}{=} (A_t, B_t, C_t)$,

$A_t =d\alpha -  it [<\psi_1, \beta_1>_H +  < \psi_2, \beta_2>_H ] \omega $,

$B_t = \partial \gamma^{0,1}  \in \Omega^{2}(M,{\CC})$

$C_t=  \left[ \begin{array}{cc}
  \bar{\partial} + A^{0,1} & \frac{t}{2} \bar{\phi} d \bar{z}  \\
 -\frac{t}{2} \bar{\phi} d \bar{z} & \bar{\partial} + A^{0,1}
\end{array} \right] \left[
\begin{array}{c}
\beta_1 \\
\bar{\beta}_2
\end{array} \right]    +  t\left[ \begin{array}{cc}
 \alpha^{0,1}  & \frac{1}{2}\gamma^{0,1}   \\
 -\frac{1}{2} \gamma^{0,1} & \alpha^{0,1} 
\end{array} \right] \left[
\begin{array}{cc}
\psi_1 \\
\bar{\psi}_2
\end{array} \right] \in V.$

Clearly, $ind ({\mathcal C}^t)$ does not depend on $t$. The complex
${\mathcal C}^0$ (for $t=0$ ) is
\begin{eqnarray*}
 0 \rightarrow \Omega^0(M, i{\RR })\stackrel{d^{\prime}_1}{\rightarrow}
\Omega^1 (M, i{\RR }) \oplus \Gamma (M, {\mathcal L}) \oplus {\mathcal H}
\stackrel{d^{\prime}_2}{\rightarrow} \Omega^2 (M, i{\RR }) \\
\oplus \Omega^{2}(M, {\CC}) \oplus V \rightarrow 0
\end{eqnarray*}
where $$d_1^{\prime} f = (df, 0, 0),$$ 
$$d_2^{\prime} (\alpha, \beta, \gamma)  = (d\alpha, \partial \gamma^{0,1},{\mathcal D}_A \beta).$$ 

 Here  ${\mathcal D}_A =   \left[ \begin{array}{cc}
\bar{\partial} + A^{0,1} & 0     \\
0 & \bar{\partial} + A^{0,1}
\end{array} \right]. $

${\mathcal C}^0$ decomposes into a direct sum of three complexes

$(a)$ $ 0 \rightarrow \Omega^0(X, i{\RR }) \stackrel{ d}{\rightarrow} \Omega^1 (M, i{\RR }) \stackrel{d}{\rightarrow}  \Omega^2 (M, i{\RR }) \rightarrow 0,$

$(b)$   $ 0 \rightarrow \Omega^{1}(M, i {\RR}) \stackrel{\partial}{\rightarrow} \Omega^{1,1} (M,i{\RR}) \rightarrow 0 $

$(c)$  $ 0 \rightarrow \Gamma (M,S) \stackrel{{\mathcal
D}_A}{\rightarrow} \Gamma (M,S^{\prime}) \rightarrow 0,$ where $S=
L \oplus L$, $S^{\prime} = (L \otimes K) \oplus (L \otimes K)$.

dim $H^1$(complex (a))$=2g$, dim $H^1$(complex (b))$=2g$.

The complex $(c)$ breaks into two complexes as follows

  $ (c1) $ $ 0 \rightarrow \Gamma (M,L)
\stackrel{\bar{\partial} + A^{0,1}}{\rightarrow} \Gamma (M, L \otimes K)
\rightarrow 0. $

$ (c2) $ $ 0 \rightarrow \Gamma (M, L)
\stackrel{\bar{\partial} + A^{0,1}}{\rightarrow}\Gamma (M, L \otimes K) \rightarrow 0. $

$(c1)$ comes from the equation $(\bar{\partial}  + A^{0,1}) \beta_1 =0$ 
and $(c2)$ is the equation 
$(\bar{\partial} + A^{0,1}) \bar{\beta}_2=0,$ which is holomorhicity of the
sections $\beta_1$ and $\bar{\beta}_2$ of $L$. 

By Riemann Roch, the index of $(c1)$ is $(c_1 ( L) - g + 1)$ and
that of $(c2)$ is $(c_1 (L) - g + 1) $ and  thus the sum
is $2g + 2g + 2c_1(L) -2g + 2 $ or $2g+ 2 c_1(L) +  2$.

If $\Phi = 0$ then case $(i)$ if $\psi_1$ and $\bar{\psi}_2$   are not identicaly zero,then the dimension is $2 c_1 (L) + 2 $ since complex (b) is missing 
 case $(ii)$ if $\psi_1 \equiv 0$ then
the dimension is $g + c_1(L) +1 $ since complex $ (b)$  and complex $(c1)$ are missing. 

\end{proof}

\section{Family of  symplectic  structures}

In the next section we discuss a standard symplection form and a variation of it which gives a whole family of symplectic structures.

Let ${\mathcal C} = {\mathcal A} \times \Gamma (M, L \oplus L)
\times {\mathcal H}$ be the space on which equations $(2.1) -
(2.3)$ are imposed. 

 Let $p= (A, \Psi, \Phi) \in {\mathcal C}$, $X
= ( \alpha_1, \beta, \gamma_1)$, $Y= (\alpha_2, \eta, \gamma_2)$
$\in T_p {\mathcal C}$. 

Let us define 
$<\beta, \eta>_H = \beta_1 H \bar{\eta}_1 + \bar{\beta}_2 H \eta_2.$

Let  $*: \Omega^{1} \rightarrow \Omega^{1}$ is  the Hodge star
operator on $M$ which acts as follows:  $*(\alpha^{1,0}) = -i \alpha^{1,0}$, and $ *(\alpha^{0,1})= i \alpha^{0,1}$. 

On ${\mathcal C}$ one can define a metric
\begin{eqnarray*}
 g( X, Y) = \int_M *\alpha_1 \wedge \alpha_2  +  \int_M Re < \beta, \eta>_H \omega  + \int_M * \gamma_1 \wedge \gamma_2
\end{eqnarray*}

Note: if we take $\alpha_1 = a dz - \bar{a} d \bar{z} $ and $\gamma_1 = c  dz - \bar{c} d \bar{z}$ then it is easy to check that
$$g(X, X) =  4 \int_M |a|^2 dx \wedge dy + \int_M (|\beta_1|^2 + |\beta_2|^2 ) dx \wedge dy +  4 \int_M |c|^2 dx \wedge dy  $$ 
which is of definite sign.

Define an almost complex structure  ${\mathcal I} = \left[
\begin{array}{cccc}
* & 0  & 0 & 0\\
0 & i & 0  & 0\\
0 & 0  & i & 0\\
0 & 0 & 0 & *
\end{array} \right] : T_p {\mathcal C} \rightarrow T_p {\mathcal C}$
We define
\begin{eqnarray*}
\Omega(X, Y) = - \int_{M} \alpha_1 \wedge \alpha_2 + \int_{M} Re <
I \beta , \eta> \omega - \int_M \gamma_1 \wedge \gamma_2
\end{eqnarray*}
where $ I = \left[\begin{array}{cc}
i & 0 \\
0  & i
\end{array} \right]$
such that $ g ({\mathcal I} X, Y) = \Omega ( X, Y).$
Moreover, we have the following:

\begin{proposition} The metrics $g$,  the symplectic form $\Omega$,
and the almost complex structure  ${\mathcal I}$ are invariant
under the gauge group action on ${\mathcal C}$. \label{inv}
\end{proposition}

\begin{proof}
 Let $p = (A, \tilde{\Psi}, \Phi) \in {\mathcal C}$ and $u \in G, $ where
$u \cdot p = (A + u^{-1} du, u^{-1} \Psi, \Phi)$.

Then $u_* : T_p {\mathcal C} \rightarrow T_ {u \cdot p} {\mathcal C} $ is given
by the mapping $(Id, u^{-1}, Id)$ and it is now easy to check that $g$ and
$\Omega$ are invariant and ${\mathcal I}$ commutes with $u_*$.
\end{proof}

\begin{proposition}
\label{props} The equation $(2.1)$  can be realised as a moment
map $\mu = 0$ with respect to the action of the gauge group and
the symplectic form $\Omega$. \label{moment1}
\end{proposition}

\begin{proof}
Let $\zeta \in \Omega(M, i{\RR })  $ be the Lie algebra of the
gauge group (the gauge group element being $u = e^{ \zeta}$ );
It generates a vector field $X_{\zeta}$ on ${\mathcal C}$ as follows :
$$X_{\zeta} (A, \Psi, \Phi) = (d \zeta, -\zeta \Psi, 0) \in T_p
{\mathcal C},p = (A, \Psi, \Phi) \in {\mathcal C}.$$

We show next that $X_{\zeta}$ is Hamiltonian. Namely, define
$H_{\zeta} : {\mathcal C} \rightarrow {\CC} $ as follows: $$
H_{\zeta} (p) = \int_{M} \zeta \cdot (  F_{A} - i\frac{( |\psi_1|_H
^2 + |\psi_2|_H ^2 )}{2} \omega). $$  Then for $X = (\alpha,
\beta, \gamma) \in T_p {\mathcal C}$.
 \begin{eqnarray*}
 dH_{\zeta} ( X ) & = & \int_M \zeta d \alpha  -i \int_M \zeta  Re ( \psi_1 H \bar{\beta_1} + \psi_2 H  \bar{ \beta_2} )  \omega   \\
 &= &\int_M (-d \zeta) \wedge \alpha  + \int_M Re < I (-\zeta
\left[ \begin{array}{cc}
\psi_1 \\
\bar{\psi}_2
\end{array} \right]) , \left[ \begin{array}{cc}
\beta_1 \\
\bar{\beta}_2
\end{array} \right] >_H \omega \\
 & = & \Omega ( X_{\zeta},  X ),
 \end{eqnarray*}
where we use that $\bar{\zeta} = - \zeta$.

 Thus we can define the moment map $ \mu : {\mathcal C} \rightarrow
 \Omega^2 ( M, i{\RR} )= {\mathcal G}^* $ ( the dual of the Lie
 algebra of the gauge group)  to be $$ \mu ( A, \Psi)
 \stackrel{\cdot}{=} (F(A) - i\frac{( | \psi_1 |_H ^2 +
 |\psi_2|_H^2)}{2}  \omega). $$ Thus equation $(2.1))$ is $\mu = 0$.
 \end{proof}

 \begin{lemma}
 Let $S$ be the solution spaces to equation $(2.1)- (2.3)$,  $X \in
 $  $ T_p {\mathcal S}$. Then ${\mathcal I}X $ $\in T_p {\mathcal S}$
 if and only if $X$ is orthogonal to the gauge orbit $ O_p = G
 \cdot p$. \label{ortho}
 \end{lemma}

 \begin{proof}
 Let $X_{\zeta} \in T_p O_p,$ where $\zeta $ $\in $ $\Omega^{0} (M,
i {\RR })$,
$g( X , X_{\zeta} ) = -\Omega ({\mathcal I} X, X_{\zeta} ) =
- \int_M \zeta \cdot d \mu ( {\mathcal I} X ),$ and therefore
${\mathcal I} X$ satisfies the linearization of equation $(2.1)$ iff
$ d \mu ({\mathcal I} X)  = 0$, i.e.,  iff $g (X, X_{\zeta}) = 0$ for all
$\zeta$. Second, it is easy to check that ${\mathcal I} X$ satisfies the
 linearization of equation $(2.2), (2.3) $ whenever $X$ does.

For instance the action of ${\mathcal I}$ in the linearisation of 
equation $(2.3)$ is 

$ \left[ \begin{array}{cc}
  \bar{\partial} + A^{0,1} & \frac{1}{2} \bar{\phi} d \bar{z}  \\
 -\frac{1}{2} \bar{\phi} d \bar{z} & \bar{\partial} + A^{0,1}
\end{array} \right] \left[
\begin{array}{c}
i \beta_1 \\
i \bar{\beta}_2
\end{array} \right]    +  \left[ \begin{array}{cc}
 i \alpha^{0,1}  & \frac{i}{2}\gamma^{0,1}   \\
 -\frac{i}{2} \gamma^{0,1} & i \alpha^{0,1} 
\end{array} \right] \left[
\begin{array}{cc}
\psi_1 \\
\bar{\psi}_2
\end{array} \right] = 0 $ since the factor of $i$ comes out and the 
remaining equation is linearization of $(2.3)$.  
(Note that the action of ${\mathcal I}$ is  $\beta_1 \rightarrow i \beta_1,$
$\bar{\beta}_2 \rightarrow i \bar{\beta}_2$,   $ \alpha^{0,1} \rightarrow i \alpha^{0,1}$ and 
 $\gamma^{0,1} \rightarrow i \gamma^{0,1}.$)  
 \end{proof}

 \begin{theorem}
 ${\mathcal N} $  has a natural symplectic structure and an almost
 complex structure  compatible with the symplectic form $\Omega $
 and the metric $g$. \label{alcom}
 \end{theorem}

 \begin{proof}
 First we show that the almost complex structure descends to
 ${\mathcal N}$. Then using this and the symplectic quotient
 construction we will show that $\Omega$ gives a symplectic
 structure on ${\mathcal N}$.

 (a) To show that ${\mathcal I}$ descends as  an almost complex
 structure we let $pr: {\mathcal S} \rightarrow {\mathcal S}/G =
 {\mathcal N}$ be the projection map and set $[p] = pr (p)$. Then
 we can naturally identify $T_{[p]}  {\mathcal N} $ with the
 quotient space $T_p {\mathcal S} / T_p O_p, $ where $ O_p = G
 \cdot p $ is the gauge orbit. Using the metric $g$ on ${\mathcal
 S}$ we can realize $T_{[p]} {\mathcal N}$ as a subspace in $T_p
 {\mathcal S}$ orthogonal to  $T_p O_p$. Then by lemma
 ~\ref{ortho}, this subspace is invariant under ${\mathcal I}$.
Thus $I_{[p]} ={\mathcal I} |_{T_p (O_p )^{\perp}}$, gives the desired
almost  complex structure. This construction does not depend on the choice of
$p$ since ${\mathcal I}$ is $G$-invariant.

 (b) The symplectic structure $\Omega$ descends to $\mu^{-1}(0) /
 G$, (by proposition \ref{props} and by the Marsden-Wienstein
 symplectic quotient construction ,~\cite{GS}, ~\cite{H}, since the
 leaves of the characteristic foliation are the gauge orbits). Now,
 as a $2$-form $\Omega$ descends to ${\mathcal N}$,   due to
 proposition (~\ref{inv}) so does the metric $g$. We check that
 equation $(2.2), (2.3),$ does not give rise to new degeneracy of
 $\Omega$ (i.e. the only degeneracy of $\Omega$ is due to $(2.1)$
 but along gauge orbits). Thus $\Omega $ is symplectic on ${\mathcal N}$.
 Since $g$ and ${\mathcal I}$ descend to ${\mathcal N}$ the latter is
 symplectic and almost complex.
 \end{proof}

Choose a  ${\psi_0} \in \Gamma (M, L) $ such that its gauge equivalence class is $fixed$ and  $\psi_0 = 0$ only 
on a set of measure zero on $M$. This $\psi_0$ has nothing to do with $\psi_1$ 
but we allow it to gauge transform as $u^{-1} \psi_0$ when $\psi_1$ gauge transforms to $u^{-1} \psi_1.$ (This will be handy in defining the determinant line bundles). 

Define a symplectic form on ${\mathcal C}$ as
\begin{eqnarray*}
\Omega_{\psi_0}(X, Y) &=& -\int_{M} \alpha_1 \wedge \alpha_2 + \int_{M} Re <
I \beta , \eta>_H  |\psi_0|_H^2 \omega \\
& & - \int_M \gamma_1 \wedge \gamma_2 \\
&=& -\int_{M} \alpha_1 \wedge \alpha_2 + \frac{i}{2} \int_{M}[( \beta_1 H \bar{\eta}_1 - \bar{\beta}_1 H \eta_1 ) \\
& & - ( \beta_2 H \bar{\eta}_2 - \bar{\beta}_2 H \eta_2 )] |\psi_0|^2_H \omega - \int_M \gamma_1 \wedge \gamma_2
\end{eqnarray*}

$|\psi_0|^2_H$ plays the role of  a conformal rescaling of the volume form 
$\omega$ on $M$ 
which appears in $\Omega$, where we allow the conformal factor to have zeroes 
on sets of measure zero. 

\begin{theorem}
$\Omega_{\psi_0}$ descends to ${\mathcal M}$ as a symplectic form.
\end{theorem}

\begin{proof}
Let $p= (A, \Psi, \Phi).$

It is easy to show that $\Omega_{\psi_0}$ is closed (this follows from the fact that on ${\mathcal C}$ it is a constant form -- does not depend on $(A, \Psi, \Phi)$). 
We have to show it is 
non-degenerate. 

Suppose there exists $(\alpha_1, \beta, \gamma_1) \in T_{[p]}({\mathcal N})$
s.t. $$\Omega_{\Psi_0} ((\alpha_2, \eta, \gamma_2), (\alpha_1, \beta, \gamma_1)) = 0$$ 
$\forall$ $(\alpha_2, \eta, \gamma_2) \in T_{[p]} ({\mathcal N})$.  
Using the metric ${\mathcal G}$  we identify $T_{[p]} {\mathcal N}$ with 
the subspace in $T_p {\mathcal S},$  ${\mathcal G}$-orthogonal to  
$T_p O_p$ (i.e. the tangent space to the moduli space is identified to the tangent space to solutions which are orthogonal to the gauge orbits, the orthogonality is 
with respect to the metric ${\mathcal G}$.)
Thus $(\alpha_1, \beta, \gamma_1), (\alpha_2, \eta, \gamma_1)$ satisfy the linearization of equation $(2.1)$, $(2.2)$ and $(2.3)$ and  ${\mathcal G} ((\alpha_1, \beta, \gamma_1), X_{\zeta}) = 0 $ and 
 ${\mathcal G} ((\alpha_2, \eta, \gamma_1), X_{\zeta}) = 0$ for all $\zeta$. 

Now, by ~\ref{ortho}, ${\mathcal I} (\alpha_1, \eta, \gamma_1) \in T_{p} S.$ Also, 
\begin{eqnarray*}
{\mathcal G}({\mathcal I} (\alpha_1, \beta, \gamma_1), X_{\zeta}) &=&  
\Omega ((\alpha_1, \beta, \gamma_1), X_{\zeta}) \\ 
&=& -\int_M \zeta  d \mu ((\alpha_1, \beta, \gamma_1))\\
&=& 0
\end{eqnarray*}
since $d \mu ((\alpha_1, \beta, \gamma_1))= 0$ is precisely one of the equations
saying  that $(\alpha_1, \beta, \gamma_1) \in T_p S$.
Thus ${\mathcal I} (\alpha_1, \beta, \gamma_1) \in T_{[p]} {\mathcal N},$ 
(since it is in $T_p S$ and ${\mathcal G}$-orthogonal to gauge orbits).

Take $ (\alpha_2, \eta, \gamma_1) = {\mathcal I} (\alpha_1, \beta, \gamma_1) = (* \alpha_1, I \beta, * \gamma_1). $ Then
\begin{eqnarray*}
0 &=&  \Omega_{\psi_0} ({\mathcal I} (\alpha_1, \beta, \gamma_1),  (\alpha_1, \beta, \gamma_1)) \\ 
&=& -\int_M ( * \alpha_1 \wedge \alpha_1) + \int_M Re < I(I \beta),  \beta>_H |\psi_0|_H^2 \omega - \int_M (* \gamma_1 \wedge \gamma_1)  \\
&=& -2i \int_M  |a|^2 dz \wedge d \bar{z}  - i \int_M (| \beta_1|_H^2 + 
|\beta_2|_H^2) |\psi_0|_H^2 e^{2 \sigma} h^2 dz \wedge d \bar{z} \\
& & - 2i \int_M |c|^2 dz \wedge d \bar{z} 
 \end{eqnarray*}
where $\omega = i e^{2 \sigma} h^2 dz \wedge d\bar{z} $ and 
$\alpha_1 = a dz - \bar{a} d \bar{z}  \in \Omega^1(M, i \RR)$ 
and $*\alpha_1 = -i( a dz + \bar{a} d \bar{z} )$ and $\gamma_1 = c dz - \bar{c} d \bar{z}. $ 
By the same sign of all the terms and the fact that $\psi_0$ has zero on a set of measure zero on $M$,   $(\alpha_1, \beta, \gamma_1) = 0$ a.e. 
Thus $\Omega_{\psi_0} $ is symplectic. 
\end{proof}

\section{Prequantum line bundle}
 In this section we briefly review the Quillen construction of the determinant 
line bundle of the Cauchy Riemann operator  $\bar{\partial}_A = \bar{\partial} + A^{(0,1)}$, ~\cite{Q},
which  enables  us to construct prequantum line bundle on the moduli 
space ${\mathcal N}$.

First let us note that a connection $A$ on a $U(1)$-principal bundle induces 
a  connection on any associated line bundle $L$. 
We will denote this connection also by $A$ since  the same ``
Lie-algebra valued $1$-form'' $A$ (modulo representations)  gives  a covariant 
derivative operator enabling you to take derivatives of  sections of $L$ 
~\cite{N}, page 348.
A very clear description of
the determinant line bundle can be found in ~\cite{Q} and
~\cite{BF}. Here we mention  the formula for the Quillen curvature
of the determinant line bundle $\wedge^{\rm top} (Ker \bar{\partial}_A)^{*}
\otimes \wedge^{\rm top}(Coker \bar{\partial}_A) = {\rm det}(\bar{\partial}_A)$,
 given the canonical unitary connection $\nabla_Q$, induced by the Quillen 
metric,~\cite{Q}.
Recall that the affine space ${\mathcal A}$ (notation as
in ~\cite{Q}) is an infinite-dimensional K\"{a}hler manifold. Here
each connection  is identified with its $(0,1)$ part which is the holomorphic 
part. Since the connection $A$ is unitary (i.e. $A = A^{(1,0)} + A^{(0,1)}$  
s.t. $\overline{A^{(1,0)}} = -A^{(0,1)}$) this identification is easy.
 In fact, for every $A \in {\mathcal A}$,
$T_A^{\prime} ({\mathcal A}) \stackrel{~}{=} \Omega^{0,1} (M, i \RR)$ and 
the corresponding K\"{a}hler  form  is given by 
\begin{eqnarray*}
F(\alpha_1^{(0,1)}, \alpha_2^{(0,1)}) &=&  {\rm Re} \int_M  (\alpha_1^{(0,1)} \wedge *_1 \alpha_2^{(0,1)}),\\
&=& - \frac{1}{2} \int_M \alpha_1 \wedge \alpha_2
\end{eqnarray*}
where $\alpha^{(0,1)}, \beta^{(0,1)} \in \Omega^{0,1} (M, i\RR),$ $\alpha_i = \alpha_i ^{1,0} + \alpha_i^{0,1} $  and 
$ *_1 $ is the Hodge-star operator such that 

$*_1(\alpha^{1,0}) = - \overline{\alpha^{1,0}} = \alpha^{0,1}$ and 

$*_1 (\alpha^{0,1}) =  \overline{\alpha^{0,1}} = - \alpha^{1,0}$ where  we have used
$\overline{\alpha_i^{(0,1)}} = - \alpha_i^{(1,0)}$, $i =1,2$. 
Let $\nabla_Q$ be the conection induced from the Quillen metric. Then the 
Quillen curvature of ${\rm det} (\bar{\partial}_A)$ is 
\begin{eqnarray*}
{\mathcal F}(\nabla_Q) &=& \frac{i}{ \pi} F \\
&=&\frac{-i}{2 \pi} \int_M(\alpha_1 \wedge \alpha_2).
\end{eqnarray*}

\section{Prequantum bundle on the moduli space ${\mathcal N}$}

First we note that to the connection $A$ we can add any one form and
still obtain a covariant derivative operator. 

Let $\omega = i e^{2 \sigma} h^2 dz \wedge d \bar{z}$ where recall $h$ is real. 
Let $\theta = h dz$ , $\bar{\theta} = h d \bar{z} $ be 1-forms (~\cite{GH}, page 28)
 such that $ \omega=i \theta \wedge \bar{\theta} = i e^{2 \sigma} h^2 dz \wedge d \bar{z}$.
Let $\psi_0$ be the same  section used to define $\Omega_{\psi_0}$ whose gauge equivalence class is fixed, and which gauge transforms in the same way as $\psi_1$ and $\bar{\psi}_2$.

$\psi_0$ has zero on a set of measure zero on $M$.
Note $ \psi_1 H \bar{\psi}_0$ and $\psi_2 H \psi_0$ are smooth gauge invariant functions on $M$.
Thus  we define $$B_{\pm} = B_{\pm}^{0,1} + B_{\pm}^{1,0}$$
such that $$B_{\pm}^{0,1} = \pm \bar{\psi_2} H \bar{\psi}_0 \bar{\theta} - \psi_1 H \bar{\psi}_0 \bar{\theta}, $$ 
$$B^{1,0}_{\pm}=  \bar{\psi}_1 H \psi_0 \theta \mp \psi_2 H \psi_0 \theta $$

  $B_{\pm}$  are two  unitary $1$-forms we would like
to add to the connection $A$ to make another connection form. (Note that $B^{0,1} = - \overline{B^{1,0}},$ as apt for unitary $1$-forms. ) 
Note that $B$ is gauge invariant, since $\psi_1,$ $\bar{\psi}_2$ and $\psi_0$ 
gauge transform in the same way.   Note that $A^{(0,1)} \pm B^{(0,1)}$ are the $(0,1)$
parts of  a connection defined by $A \pm B  = A^{(0,1)} \pm B^{(0,1)} + 
A^{(1,0)} \pm B^{(1,0)},$ where $B$ can be one of $B_{\pm}$.

{\bf Definitions:} Let us denote by ${\mathcal L}_{1}^{\pm} = 
{\rm det} [ \frac{1}{\sqrt{4}} (\bar{\partial}+  A^{(0,1)} \pm B_{+}^{(0,1)})]$
 two determinant bundles on the affine spaces ${\mathcal J}_{\pm} = \{\frac{1}{\sqrt{4}} (A^{(0,1)} \pm  B_{+}^{0,1}) | A \in {\mathcal A}, \Psi \in \Gamma(M, L \oplus L ) \}$ respectively. These affine spaces  are isomorphic to 
$ {\mathcal A} \times \Gamma(M,L \oplus L ) \times \Phi_0$, $\Phi_0$ being a fixed Higgs field.  We can extend it to all of ${\mathcal C} =  {\mathcal A} \times \Gamma(M,L \oplus L ) \times {\mathcal H}$ by defining the fibers to be same for all $\Phi$. 

Similarly define ${\mathcal L}_{2}^{\pm} = 
{\rm det} [ \frac{1}{\sqrt{4}} (\bar{\partial}+  A^{(0,1)} \pm B_{-}^{(0,1)})]$

Thus ${\mathcal P}_{\psi_0} = {\mathcal L}^{+}_{1} \otimes {\mathcal L}^{-}_{1}  \otimes {\mathcal L}^{+}_{2} \otimes {\mathcal L}^{+}_{2}$
well-defined line bundle on ${\mathcal C}$.

\begin{lemma}
${\mathcal P}_{\psi_0}$ is a well-defined line bundle over 
${\mathcal N} \subset {\mathcal C}/G$, 
where $G$ is the gauge group.
\end{lemma}
\begin{proof}
First consider the Cauchy-Riemann operators 
$ D=  \frac{1}{\sqrt{4}} (\bar{\partial} +  A^{(0,1)} + B_{+}^{(0,1)})$. Under gauge transformation 
$D=[\frac{1}{\sqrt{4}} (\bar{\partial} + A^{(0,1)}  + B_{+}^{(0,1)})]
\rightarrow D_g= g[\frac{1}{\sqrt{4}} (\bar{\partial} +  A^{(0,1)} + B_{+}^{(0,1)})]g^{-1} $.
We can show that the operators $D$ and $D_g$ have isomorphic 
kernel and cokernel and their corresponding Laplacians have the 
same spectrum and the eigenspaces are of the same dimension. Let 
$\Delta$ denote the Laplacian corresponding to $D$ and $\Delta_g$ 
that corresponding to $D_g$. 
The Laplacian is $\Delta = \tilde{D} D$ where 
$\tilde{D} = [\frac{1}{\sqrt{4}} (\partial + A^{(1,0)} + B_{+}^{(1,0)})]$, where recall $\overline{A^{(1,0)}} = -A^{(0,1)}$ and $\overline{B_{+}^{(1,0)}} = - B_{+}^{(0,1)}$. Note that  
$\tilde{D} \rightarrow  \tilde{D}_g = g \tilde{D} g^{-1}$ under gauge transformation. Then $\Delta_g = g \Delta g^{-1}$. 
Thus the isomorphism of eigenspaces is  $s \rightarrow g s$. We describe here how to define the line bundle on the moduli space.
Let $K^a(\Delta)$ be the direct sum of 
eigenspaces of the operator $\Delta$ of 
eigenvalues $< a$, over the open subset 
$U^a = \{ \frac{1}{\sqrt{4}}(A^{(0,1)} + B_{+}^{(0,1)}) | a \notin {\rm Spec} \Delta \}$ of the affine 
space ${\mathcal J_{+}}.$ The determinant line bundle is defined using the exact sequence
$$ 0 \rightarrow {\rm Ker} D \rightarrow K^a(\Delta) \rightarrow 
D(K^a(\Delta)) \rightarrow {\rm Coker} D \rightarrow 0$$ 
Thus 
one identifies 

$\wedge^{{\rm top} }({\rm Ker} D)^* \otimes \wedge^{{\rm top} }
({\rm Coker} D)$ with 
 $\wedge^{{\rm top}}(K^a(\Delta))^* \otimes \wedge^{{\rm top}} 
(D(K^a(\Delta)))$  (see ~\cite{BF},  for more details) and 
there is an isomorphism of the fibers as $D \rightarrow D_g$. 
Thus one can identify 

$$ \wedge^{{\rm top}}(K^a(\Delta))^* \otimes \wedge^{{\rm top}} 
(D(K^{a}(\Delta))) \equiv
\wedge^{{\rm top}}(K^a(\Delta_g))^* \otimes \wedge^{{\rm top}} 
(D(K^{a}(\Delta_g))).$$
By extending this definition from 
$U^a$ to $V^a = \{(A, \Psi, \Phi)| a \notin {\rm Spec} \Delta \}$, 
an open subset of ${\mathcal C}$,  we can define the fiber over 
the quotient space $V^a/G$ to be the 
equivalence class of this fiber. Covering ${\mathcal C}$ with open sets of the 
type $V^a$, we can define it on ${\mathcal C}/G$. Then we can restrict it to
${\mathcal N} \subset {\mathcal C}/G$.

Similarly one can deal with the other cases of  
$[\frac{1}{\sqrt{4}} (\bar{\partial} + A^{(0,1)} \pm B_{\pm}^{(0,1)})]$.
For instance, let $([A], [\Psi], [\Phi]) \in {\mathcal C}/G,$ 
where $[A], [\Psi], [\Phi]$ are gauge equivalence classes of $A, \Psi, \Phi$, 
respectively.  Then associated to the equivalence class $([A], [\Psi], [\Phi])$ in the 
base space, there is an 
equivalence class of fibers coming from the identifications 
of ${\rm det} [\frac{1}{\sqrt{4}}(\bar{\partial} + A^{(0,1)} - B_{+}^{(0,1)})]$ with ${\rm det}[g(\frac{1}{\sqrt{2}}(\bar{\partial} + A^{(0,1)} - B_{+}^{(0,1)}))g^{-1}]$ as mentioned in the previous case. 

 This way one can prove that  ${\mathcal P}_{\Psi_0}$ is well defined 
on ${\mathcal C}/G$. Then we restrict it to 
${\mathcal N} \subset {\mathcal C}/G$.
\end{proof}

Next, in a similar way, we define two other determinant line bundles. Recall
$\overline{\Phi^{(1,0)}} = - \Phi^{(0,1)}.$ Let us denote by ${\mathcal M}_{\pm} = 
{\rm det} [ \frac{1}{\sqrt{2}} (\bar{\partial}+  A^{(0,1)}) \pm \Phi^{(0,1)}]$ a determinant bundle on ${\mathcal J}_{\pm} = \{\frac{1}{\sqrt{2}} (A^{(0,1)}) \pm 
 \Phi^{(0,1)} | A \in {\mathcal A}, \Phi \in {\mathcal H} \}$ which is isomorphic to 
${\mathcal A} \times {\mathcal H}$. We can extend it to ${\mathcal C} = 
{\mathcal A} \times \Gamma(M, L \oplus L) \times {\mathcal H}$ by defining the fibers  to be the same  for all $\Psi$. 
Thus ${\mathcal M} = {\mathcal M}_{+} \otimes {\mathcal M}_{-}$
well-defined line bundle on ${\mathcal C}$.

This can be defined exactly in a similar way to ${\mathcal P}_{\Psi_0}$ over the moduli space ${\mathcal N}$.

[Note: The square root of $2$ comes with the $\bar{\partial}+  A^{(0,1)}$-term alone.]

{\bf Curvature and symplectic form:}

Let $p = (A, \Psi, \Phi) \in S$. Let $X, Y \in T_{[p]}{\mathcal N}$. 
Since $T_{[p]}{\mathcal N}$ can be identified with a subspace in 
$T_p S$ orthogonal to $T_p O_p$, if we write 
$X =(\alpha_1, \beta, \gamma_1)$ and $Y=(\alpha_2, \eta, \gamma_2)$, 
(notation as before)
 then 
$X,Y$ can be said to satisfy a) $X, Y \in T_p S$ and b) 
$X, Y$ are ${\mathcal G}$-orthogonal to $T_p O_p $, the tangent space to the gauge orbit.  

Let ${\mathcal F}_{{\mathcal L}_{\pm}}$
 denote the Quillen curvatures of the four
determinant line bundles ${\mathcal L}_{1}^{\pm}$, ${\mathcal L}_2^{\pm}$,  
 respectively, which are determinants  of Cauchy-Riemann operators of the connections $  \frac{1}{\sqrt{4}}(A^{(0,1)}\pm B^{(0,1)}_{\pm}).$  In the curvature formula of Quillen the terms that will appear are 
$\frac{1}{\sqrt{4}}(\alpha_1 \pm b_{\pm})$ and $\frac{1}{\sqrt{4}}(\alpha_2 \pm c_{\pm})$ where 
$b_{\pm} = b_{\pm}^{(1,0)} + b_{\pm}^{(0,1)},$  $c_{\pm} = c_{\pm}^{(1,0)} + c_{\pm}^{(0,1)}$ such that
$$b^{(0,1)}_{\pm}  =  \pm \bar{\beta}_2 H \bar{\psi}_0 \bar{\theta} - \beta_1 H \bar{\psi}_0 \bar{\theta}  $$

$$b^{(1,0)}_{\pm} = \bar{\beta}_1 H \psi_0 \theta \mp \beta_2 H \psi_0 \theta $$

$$c^{(0,1)}_{\pm}  =  \pm \bar{\eta}_2 H \bar{\psi}_0 \bar{\theta} - \eta_1 H \bar{\psi}_0 \bar{\theta}     $$  
$$c^{(1,0)}_{\pm}  = \bar{\eta}_1 H \psi_0 \theta \mp \eta_2 H \psi_0 \theta  $$

\begin{eqnarray*}
{\mathcal F}_{{\mathcal L}^{\pm}_{1}} (X,Y)&=& 
-\frac{i}{2\pi}  \int_M \frac{1}{\sqrt{4}}(\alpha_1 \pm  b_{+}) \wedge \frac{1}{\sqrt{4}}(\alpha_2 \pm c_{+}) \\
&=& -\frac{i}{8\pi} \int_M [(\alpha_1 \wedge \alpha_2 ) \pm (b_{+} \wedge \alpha_2) \\
& & \pm (\alpha_1 \wedge c_{+}) + (b_{+} \wedge c_{+})]
\end{eqnarray*}

 \begin{eqnarray*}
{\mathcal F}_{{\mathcal L}^{\pm}_{2}} (X,Y)&=& 
-\frac{i}{2\pi}  \int_M \frac{1}{\sqrt{4}}(\alpha_1 \pm  b_{-}) \wedge \frac{1}{\sqrt{4}}(\alpha_2 \pm c_{-}) \\
&=& -\frac{i}{8\pi} \int_M [(\alpha_1 \wedge \alpha_2 ) \pm (b_{-} \wedge \alpha_2) \\
& & \pm (\alpha_1 \wedge c_{-}) + (b_{-} \wedge c_{-})]
\end{eqnarray*}

One can easily compute that 
\begin{eqnarray*}
{\mathcal F}_{{\mathcal P}_{\psi_0}} (X, Y) &=&
({\mathcal F}_{{\mathcal L}^{+}_{1}} +   {\mathcal F}_{{\mathcal L}^{-}_{1}}  + {\mathcal F}_{{\mathcal L}_{2}^{+}} +   {\mathcal F}_{{\mathcal L}_{2}^{-}}    )(X,Y) \\
&=& \frac{-i}{2 \pi} [ \int_M \alpha_1 \wedge \alpha_2 + \frac{1}{2} \int_M (b_{+} \wedge c_{-} + b_{-} \wedge c_{-})] \\
&=& \frac{-i}{2\pi} \int_M [(\alpha_1 \wedge \alpha_2 ) - i [(\beta_1 H \bar{\eta}_1 - \bar{\beta}_1 H \eta_1 ) -  (\beta_2 H \bar{\eta}_2 \\
& & - \bar{\beta}_2 H \eta_2 )] |\psi_0|_H^2 \omega ]
\end{eqnarray*}
after replacing $\theta \wedge \bar{\theta} = - i \omega$.

Let ${\mathcal F}_{{\mathcal M}_{\pm}}$ denote the  curvatures of ${\mathcal M}_{\pm}.$ Then, terms like $\frac{\alpha_i}{\sqrt{2}} \pm \gamma_i$ will appear in the Quillen curvature 
formula: 
 
\begin{eqnarray*}
{\mathcal F}_{{\mathcal M}_{\pm}} (X,Y)&=& 
\frac{-i}{2\pi}  \int_M [(\frac{\alpha_1}{\sqrt{2}} \pm  \gamma_1) \wedge (\frac{\alpha_2}{\sqrt{2}} \pm \gamma_2) ]  
\end{eqnarray*}

One can easily compute that 
\begin{eqnarray*}
{\mathcal F}_{{\mathcal M}}(X, Y)
&=&({\mathcal F}_{{\mathcal M}_{+}} + {\mathcal F}_{{\mathcal M}_{-}})(X,Y)  \\
&=&  \frac{-i}{2\pi} \int_M [ (\alpha_1 \wedge \alpha_2 ) + 2(\gamma_1 \wedge \gamma_2)] \omega 
\end{eqnarray*}

{\bf Holomorphicity} Since in  $A^{0,1} \pm B^{0,1}_{\pm},$ terms with 
$\psi$ and $\bar{\psi}_2$ comes, i.e. under the action of ${\mathcal I}$, $ \alpha_1^{0,1} \pm b_{\pm}^{0,1}$ goes to $i (  \alpha_1^{0,1} \pm b_{\pm}^{0,1})$, and $\alpha^{0,1} \pm \gamma^{0,1}$ goes to $i (\alpha^{0,1} \pm \gamma^{0,1})$ these line bundles are holomorphic.

Thus, we have proven the following theorem:
\begin{theorem}
$ {\mathcal Q}_{\Psi_0} = {\mathcal P}_{\psi_0}  \otimes {\mathcal M} $ is a well-defined
holomorphic line bundle on ${\mathcal N}$ whose Quillen curvature is
$\frac{i}{\pi}\Omega_{\Psi_0}$. Thus ${\mathcal Q}_{\Psi_0}$ is a 
prequantum bundle on ${\mathcal N}$.
\end{theorem}

{\bf Remark:}

As $\psi_0$ varies, the corresponding line bundles are  all topologically
equivalent since the curvature forms have to be of integral cohomology
and that would be constant. Thus they  have the same Chern class.
Holomorphically they may differ.

\section{Alternative method for the prequantization}

We  fix the gauge equivalence class of the connection 
 $A_0$ , i.e. $A_0$ is a $fixed$ connection which 
gauge transforms like $A$ when $\Psi$ gauge transforms.

 We define two determinant line bundles on the moduli space 
in the same way as before 
${\mathcal T}_{\pm}  = {\rm det} (\bar{\partial} + A_{0}^{0,1} + B^{0,1}_{\pm})$   on  ${\mathcal N} \subset {\mathcal C}/ G$.

Let ${\mathcal T} = {\mathcal T}_{+} \otimes {\mathcal T}_{-}$ 

Then ${\mathcal F}_{{\mathcal T}_{+}}(X, Y) = \frac{-i}{2 \pi} \int_M ( b_{+} \wedge c_{+} )$ and ${\mathcal F}_{{\mathcal T}_{-}}(X,Y) = \frac{-i}{2 \pi} \int_M ( b_{-} \wedge c_{-} ).$

Thus the curvature
\begin{eqnarray*}
{ \mathcal F}_{{\mathcal T}}(X,Y) &=& {\mathcal F}_{{\mathcal T}_{+}}(X,Y) + {\mathcal F}_{{\mathcal T}_{-}}(X,Y) \\
&=&  \frac{-i}{2 \pi} \int_M (b_{+} \wedge c_{+} + b_{-} \wedge c_{-}) \\
&=& \frac{-i}{2\pi} \int_M -2i [ (\beta_1 H \bar{\eta}_1 -  \bar{\beta}_1 H \eta_1)\\ 
& & -  (\beta_2 H \bar{\eta}_2 -  \bar{\beta}_2 H \eta_2)]|\psi_0|_H^2 \omega 
\end{eqnarray*}

Define 

${\mathcal S}_{\pm} = 
{\rm det} (\bar{\partial}+  A^{(0,1)} \pm \Phi^{(0,1)})$ a determinant bundle on ${\mathcal N}$.

Let ${\mathcal S} = {\mathcal S}^2_{+} \otimes {\mathcal S}^2_{-}$.

\begin{eqnarray*}
{\mathcal F}_{{\mathcal S}_{\pm}} (X,Y)&=& 
\frac{-i}{2\pi}  \int_M [(\alpha_1 \pm  \gamma_1) \wedge (\alpha_2 \pm \gamma_2) ]  
\end{eqnarray*}

One can easily compute that 
\begin{eqnarray*}
{\mathcal F}_{{\mathcal S}} (X,Y)
&=&2 ({\mathcal F}_{{\mathcal S}_{+}} + {\mathcal F}_{{\mathcal S}_{-}})(X,Y)  \\
&=&  \frac{-i}{2\pi} \int_M [ 4(\alpha_1 \wedge \alpha_2 ) + 4(\gamma_1 \wedge \gamma_2)] \omega 
\end{eqnarray*}

It is easy to calculate that 
${\mathcal D}_{\psi_0} = {\mathcal T} \otimes {\mathcal S}$ has curvature $\frac{2i}{\pi} \Omega_{\psi_0}$.

It is also a holomorphic line bundle.

Thus we have proved
\begin{theorem} 
${\mathcal D}_{\psi_0}$ is a holomorphic prequantum line bundle on ${\mathcal N}$ with 
curvature $ \frac{2i}{\pi} \Omega_{\psi_0}$.
\end{theorem}

Harish Chandra Research Institute, Chhatnag, Jhusi, Allahabad, 211019, India.
email: rkmn@mri.ernet.in

\end{document}